\documentclass{amsart}
\usepackage[margin=1.4in]{geometry}
\usepackage{graphicx} 
\usepackage{amsmath,amsthm,amssymb,amsfonts}
\usepackage{mathtools}
\usepackage{hyperref}
\newtheorem{theorem}{Theorem}
\newtheorem{example}[theorem]{Example}
\newtheorem{lemma}[theorem]{Lemma}

\newtheorem{proposition}[theorem]{Proposition}

\newtheorem{definition}[theorem]{Definition}
\newtheorem{remark}[theorem]{Remark}
\numberwithin{equation}{section}

\newcommand{\F}{\mathbb{F}}
\newcommand{\Nsq}{{\rm N}_{q^2/q}}

\title[$2$-quasi-perfect Lee codes and abelian Ramanujan graphs]{$2$-quasi-perfect Lee codes and abelian Ramanujan graphs: a new construction and relationship}
\author[Shohei Satake]{Shohei Satake}
\address{Shohei Satake \\ Department of Applied Mathematics, Faculty of Science, Fukuoka University,  
Fukuoka,  
Japan}
\email{s.satake@fukuoka-u.ac.jp}


\begin{document}

\begin{abstract}
This paper presents a new explicit infinite family of 2-quasi-perfect $p$-ary Lee codes of length $\frac{q-1}{2}$ and dimension $\frac{q-1}{2}-2k$ for $q = p^k \ge 14$, $p\geq 5$ a prime.  
Our codes are derived from the generating set $H_q = \{(a, a^3) \mid a \in \mathbb{F}_q^*\}$ of the additive group of the finite field $\mathbb{F}_{q^2}$.  
Furthermore, we bridge between 2-quasi-perfect Lee codes constructed by Mesnager, Tang, and Qi and well-known abelian Ramanujan graphs, specifically Li's graphs and finite Euclidean graphs, providing a unified theoretical framework for these families.
\end{abstract}

\maketitle

\section{Introduction}
A {\it Lee code} is an error-correcting code under the {\it Lee distance}.
While finding perfect Lee codes is a classical problem, from a broader theoretical perspective, the focus has shifted to the construction of ``good'' codes, such as {\it quasi-perfect Lee codes}.
Constructing {\it linear} quasi-perfect Lee codes over finite fields represents a fascinating intersection of coding theory, number theory, algebraic geometry, and combinatorics.
In this context, Camarero and Mart\'{i}nez~\cite{CM2016} constructed $2$-quasi-perfect $p$-ary Lee codes of length $2\lfloor \frac{p}{4} \rfloor$ for arbitrary primes $p\equiv \pm 5 \pmod{12}$, which proves the existence of quasi-perfect Lee codes for arbitrarily large length.
For an odd prime $p$, Mesnager, Tang, and Qi~\cite{MTQ2018} formulated an equivalent condition for the existence of $2$-quasi-perfect $p$-ary linear codes of length $n$, and provided several explicit families of $2$-quasi-perfect Lee codes by generalizing the Camarero-Mart\'{i}nez construction and analyzing the solution sets of specific polynomial equations.
Moreover, the Camarero-Mart\'{i}nez and Mesnager-Tang-Qi constructions also produce Cayley graphs over the additive groups of finite fields that are {\it abelian Ramanujan graphs}.
It is worth mentioning that these graphs are optimal pseudorandom graphs in a sense of spectral gap, where their short diameters crucially affect the covering property of the corresponding Lee codes.

Nevertheless, to the best of our knowledge, constructing quasi-perfect Lee codes of arbitrarily large length remains a challenging open problem. A unified theoretical framework that rationalizes why certain algebraic structures yield such codes is notably absent; consequently, existing constructions often appear as isolated families derived from specific polynomials without clear intuitive grounding. Furthermore, while the small diameter of abelian Ramanujan graphs inherently favors the covering radius of associated Lee codes, achieving perfect or quasi-perfect status necessitates stricter structural constraints. Determining precisely which abelian Ramanujan graphs yield these optimal codes remains a significant open question.


In this paper, we first exhibit a new and simple construction for $2$-quasi-perfect linear Lee codes of arbitrarily large length.
Specifically, we investigate the $p$-ary linear code derived from the set $H_q \coloneq \{(a, a^3) \mid a \in \mathbb{F}_q^*\}$, where $p\geq 5$ is a prime and $q=p^k$, and prove that this code is a $2$-quasi-perfect code of length $\frac{q-1}{2}$ and dimension $\frac{q-1}{2}-2k$ provided that $q \geq 14$ (Theorem~\ref{thm-main}).
It is worth mentioning that this construction comes from a recent construction of abelian almost Ramanujan graphs by Forey, Fres\'{a}n, Kowalski and Wigderson~\cite{FFKW2025}, and the structure of the codes and associated Cayley graphs are different from known ones.
Moreover, we bridge the Mesnager-Qi-Tang construction of Lee codes to well-studied abelian Ramanujan graphs in combinatorics and number theory, namely,  {\it Li' s graphs} introduced by Winnie Li (\cite{LK1992}), and {\it finite Euclidean graphs} (\cite{MMST1996}). 
This provides a unified description of the families of Lee codes by the Mesnager-Qi-Tang construction (Propositions~\ref{prop-Li} and \ref{prop-euclid}).

The remainder of the paper is organized as follows. Section~\ref{sect-prelim} reviews preliminaries on finite fields, Lee codes, Cayley and Ramanujan graphs. Section~\ref{sect-main} details our construction and the proof of the $2$-quasi-perfectness employing algebraic geometry.
Section~\ref{sect-MTQ} discusses a connection between the Mesnager-Qi-Tang codes and abelian Ramanujan graphs, namely, Li's graphs and finite Euclidean graphs. 
We conclude this paper with several remarks in Section~\ref{sect-conc}.

\section{Preliminaries}
\label{sect-prelim}

\subsection{Finite fields}
Let $p$ be a prime, and $q=p^k$, where $k\geq 1$ is an integer.
Throughout this paper, $\F_q$ denotes the finite field of order $q$.
Note that $\F_p$ is the residue ring $\mathbb{Z}/(p\mathbb{Z})$, and it
is a subfield of $\F_q$.
The additive group of $\F_q$, denoted by $\F_q^+$, has the structure of the $k$-dimensional vector space over $\F_p$, and accordingly, for each integer $n\geq 1$,  $\F_q^n\coloneq \{(a_1, a_2, \ldots, a_n) \mid a_i \in \F_q, i=1,2,\ldots, n\}$ is a $kn$-dimensional vector space over the same field $\F_p$.
In addition, let $\F_q^*\coloneq \F_q \setminus \{0\}$, which is the multiplicative group of $\F_q$.
For $\ell \geq 1$, the {\it field norm} ${\rm N}_{q^\ell/q}$ from $\F_{q^\ell}^*$ to $\F_q^*$ is a surjective homomorphism defined as ${\rm N}_{q^\ell/q}(z)\coloneq z^{1+q+\cdots+q^{\ell-1}}$.
Also, for $m \geq 1$, let $GL(m,\F_q)$ denote the general linear group of rank $m$ over $\F_q$,  consisting of all invertible $m\times m$ matrices over $\F_q$.  
For further details of finite fields, see, e.g., \cite{LN1983}.

\subsection{Lee codes}
Let $p\geq 3$ be a prime.
A $p$-ary {\it (linear) Lee code} $\mathcal{C}$ of length $n$ is a linear subspace of the vector space $\F_p^n$, where the error correction is formulated via the following {\it Lee distance}: for any $\mathbf{a}, \mathbf{b} \in \F_p^n$,
\begin{align*}
    d_{L}(\mathbf{a}, \mathbf{b})\coloneq \sum_{i=1}^n\min\{|s| \mid s \equiv a_i-b_i \pmod{p},\: s \in \mathbb{Z}\}.
\end{align*}
An element of a linear code $\mathcal{C}$ is called a {\it codeword} of $\mathcal{C}$.
A code $\mathcal{C}$ is said to be {\it $t$-error-correcting} if $t$ is the largest integer such that for any $\mathbf{w} \in \F_p^n$, there exists at most one codeword $\mathbf{c} \in \mathcal{C}$ with $d_L(\mathbf{w},\mathbf{c}) \leq t$, where such $t$ is called the {\it
packing radius} of $\mathcal{C}$. A code $\mathcal{C}$ is said to be {\it $r$-covering} if $r$ is the smallest integer such that for any $\mathbf{w} \in \F_p^n$, there exists at least one codeword $\mathbf{c} \in \mathcal{C}$ with $d_L(\mathbf{w},\mathbf{c}) \leq r$, where such $r$ is called the {\it covering radius} of $\mathcal{C}$
We say a code $\mathcal{C}$ is {\it $t$-perfect} if it is both $t$-error-correcting and $t$-covering.
A famous Golomb-Welch conjecture states that a $t$-perfect code of length $n$ exists only if $t=1$ or $n=1,2$.
Instead, the best alternative is known as a {\it $t$-quasi-perfect} code, where a code $\mathcal{C}$ is said to be {\it $t$-quasi-perfect} if it is both $t$-error-correcting and $(t+1)$-covering.
Even for $t$-quasi-perfect codes with $t\geq 2$, there are few explicit constructions, and it is quite challenging to obtain explicit constructions for $t$-quasi-perfect codes.

\subsection{Cayley graphs and Ramanujan graphs}
For a finite group $G$ with identity $1$, and an inverse-closed subset $S \subset G \setminus \{1\}$ (which means that $s^{-1}\in S$ for any $s \in S$), the {\it Cayley graph} $Cay(G, S)$ on $G$ with respect to $S$ is the graph with vertex set $G$ in which two vertices $g$ and $h$ are adjacent if and only if $gh^{-1} \in S$. 
Note that $Cay(G, S)$ is an undirected $|S|$-regular graph.

The adjacency matrix $A_X$ of a simple $N$-vertex graph $X$ is an $N\times N$ $\{0,1\}$-matrix where the $(u,v)$-entry is $1$ if and only if $u$ and $v$ are adjacent in $X$, and hence, the matrix is symmetric, implying that all eigenvalues are real.
By the Perron-Frobenius theorem, for any $d$-regular graph $X$, 
all eigenvalues of $A_X$ are in the interval $[-d, d]$, where the largest eigenvalue is $d$, and $-d$ is an eigenvalue if and only if $X$ is bipartite. 
We call $\pm d$ trivial eigenvalues, and set $\lambda(X)$ as the largest absolute value of non-trivial eigenvalues of $A_X$.
A $d$-regular $N$-vertex graph $X$ is called a {\it Ramanujan graph} if  $\lambda(X) \leq 2\sqrt{d-1}$.
Similarly, we say $X$ is an {\it almost Ramanujan graph} if $\lambda(X) \leq (2+o(1))\sqrt{d-1}$ as $N \to \infty$
If $d$ is fixed, then it is well-known that $\displaystyle \liminf_{N\to \infty}\lambda(X) \geq 2\sqrt{d-1}$ by the Alon-Boppana bound (e.g.~\cite{DSV2003}), and even if $d$ grows as $N \to \infty$ (say, $\omega(1)=d<\frac{N}{2}$),  one can prove that $\lambda(X) = \Omega(\sqrt{d})$ as $N \to \infty$.
Hence, Ramanujan graphs realize the optimally small non-trivial eigenvalues in absolute value (up to constant).
It is known that (almost) Ramanujan graphs enjoy very rich properties such as high global connectivity, pseudo-randomness, and rapid mixing.
An {\it abelian Ramanujan graph} is a Cayley graph over an abelian group which is a Ramanujan graph.
For a more detailed exposition of Cayley and Ramanujan graphs, see, for instance, \cite{DSV2003}. 

\section{A generic construction}
Recall a generic construction for $2$-quasi-perfect $p$-ary Lee codes by Mesnager, Tang and Qi in \cite{MTQ2018}.

\begin{definition}[\cite{MTQ2018}]
\label{def-generic}
Let $\Gamma$ be a vector space over $\F_p$ which is spanned by a set $H=\{\pm \beta_1, \pm \beta_2, \ldots, \pm \beta_n\}$ of non-zero vectors.
Then define the $p$-ary code $\mathcal{C}(\Gamma, H)$ as 
\begin{align*}
    \mathcal{C}(\Gamma, H)\coloneq \{(c_1, \ldots,c_n) \in \F_p^n\mid c_1\beta_1+\cdots+c_n\beta_n=\mathbf{0}\}.
\end{align*}
\end{definition}

For subsets $A, B$ of a vector space, let $-A \coloneq \{-a \mid a\in A\}$, $A+B \coloneq \{a+b \mid a\in A, b\in B\}$.
For an integer $i\geq 0$, define 
\begin{align*}
    A^{(i)} \coloneq
  \begin{cases}
    \{0\} & \text{if $i=0$,} \\
    A                & \text{if $i=1$,} \\
    A^{(i-1)}+A      & \text{if $i\geq 2$.}
  \end{cases}
\end{align*}

\begin{lemma}[Proposition~6 in \cite{MTQ2018}]
\label{lem-generic}
Let $p\geq 5$ be an odd prime. 
For $\Gamma$ and $H$ satisfying the assumptions of Definition~\ref{def-generic}, we have the following.
\begin{enumerate}
    \item[$(1)$] $\mathcal{C}(\Gamma, H)$ is a $2$-perfect $p$-ary Lee code of length $n$ and dimension $n-\dim_{\F_p}(\Gamma)$ if and only if 
    $\#(\bigcup_{i=0}^{2} H^{(i)})=2n^2+2n+1=\#\Gamma$.
    
    \item[$(2)$] $\mathcal{C}(\Gamma, H)$ is a $2$-quasi-perfect $p$-ary Lee code of length $n$ and dimension $n-\dim_{\F_p}(\Gamma)$ if and only if all of the following hold:
    \begin{itemize}
        \item $\#(\bigcup_{i=0}^{2} H^{(i)})=2n^2+2n+1$,
        \item $\bigcup_{i=0}^{3} H^{(i)}=\Gamma$,
        \item $2n^2+2n+1 < \#\Gamma < \frac{1}{3}(2n+1)(2n^2+2n+3)$.
    \end{itemize}
\end{enumerate}
\end{lemma}

\begin{remark}
Consider the Cayley graph $Cay(\Gamma, H)$, where $\Gamma$ is also an abelian group. 
Then $H^{(i)}$ is the set of all vertices to which the distance from the zero-vector $\mathbf{0}$ is $i$. 
Hence, the second condition of Lemma~\ref{lem-generic} (2) shows that the diameter of $Cay(\Gamma, H)$ is $3$.
On the other hand, it is known (e.g. \cite{C1989}) that for any $d$-regular $N$-vertex graph $X$, its diameter is at most $\lceil \frac{\log(N-1)}{\log(d/\lambda(X))} \rceil$, and hence, among all $d$-regular $N$-vertex graphs, $d$-regular $N$-vertex Ramanujan graphs minimize the upper bound for the diameter. 
Thus, for constructing a $2$-quasi-perfect code $\mathcal{C}(\Gamma,H)$, it is hopeful to choose $H$ so that $Cay(\Gamma, H)$ is an abelian Ramanujan graph.
On the other hand, the conditions of Lemma~\ref{lem-generic} (2) strictly restrict the structure of $Cay(\Gamma, H)$, and hence finding suitable $H$ would be quite difficult in general.  
\end{remark}

\section{A new explicit construction of $2$-quasi-perfect Lee codes}
\label{sect-main}
For a prime $p\geq 5$ and a power $q=p^k$ with $k\geq 1$, let $H_q\coloneq \{(a, a^3) \mid a \in \F_q^*\}$.
Notice that as $q$ is odd, $\#H_q$ is even and  $\#H_q=2\cdot \frac{q-1}{2}$.
In addition, $H_q$ is the set of points on the cubic curve $y=x^3$. 

Here is a new explicit $2$-quasi-perfect Lee codes.
\begin{theorem}
\label{thm-main}
Let $p \geq 5$ be a prime. Suppose that a power $q=p^k$ with $k\geq 1$ satisfies $q\geq 14$.
Then, the code $\mathcal{C}(\F_q^2, H_q)$ is a $2$-quasi-perfect $p$-ary Lee code of length $\frac{q-1}{2}$ and dimension $\frac{q-1}{2}-2k$. 
\end{theorem}

\begin{remark}
By the Weil sum estimation (e.g.~\cite{LN1983}), $Cay(\F_{q^2}^+, H_q)$ is a $(q-1)$-regular almost Ramanujan graph.
This graph was found by Forey, Fres\'{a}n, Kowalski and Wigderson in \cite{FFKW2025}, and not isomorphic to the abelian Cayley graph associated with known $2$-quasi-perfect Lee codes~\cite{BKS2016,MTQ2018}.
\end{remark}

\begin{remark}
As shown in Theorem~\ref{thm-MTQ2018} (2), Mesnager, Qi and Tang provided a
$2$-quasi-perfect $p$-ary Lee code of length $\frac{q-1}{2}$ and dimension $\frac{q-1}{2}-2k$ for primes $p\equiv 5, 11 \pmod{12}$ and $q=p^k>12$. 
For these $p$ and $q$, the code in Theorem~\ref{thm-main} is not equivalent to their code although both have the same parameters.
Moreover, our codes exhibit new parameters for primes $p\equiv 1, 7 \pmod{12}$.
\end{remark}

To illustrate the distinct algebraic structures of our new codes, we exhibit their parity-check matrices alongside those of the Mesnager-Tang-Qi codes for several small parameters.

\begin{example}
For $p=17 \equiv 5 \pmod{12}$ and $k=1$, the parity check matrix of the code $\mathcal{C}(\F_{17}^2, H_{17})$ is
\begin{align*}
    \begin{pmatrix}
1 & 2 & 3 & 4 & 5 & 6 & 7 & 8 \\
1 & 8 & 10 & 13 & 6 & 12 & 3 & 2
\end{pmatrix}
\end{align*}
while the parity check matrix of the Mesnager-Qi-Tang code in Theorem~\ref{thm-MTQ2018} (2) is  
\begin{align*}
\begin{pmatrix}
1 & 2 & 3 & 4 & 5 & 6 & 7 & 8 \\
1 & 9 & 6 & 13 & 7 & 3 & 5 & 15
\end{pmatrix}.
\end{align*}
\end{example}

\begin{example}
For $p=19 \equiv 7 \pmod{12}$ and $k=1$, the parity check matrix of the code $\mathcal{C}(\F_{19}^2, H_{19})$ is
\begin{align*}
    \begin{pmatrix}
1 & 2 & 3 & 4 & 5 & 6 & 7 & 8 & 9 \\
1 & 8 & 8 & 7 & 11 & 7 & 1 & 18 & 7  
\end{pmatrix}.
\end{align*}
\end{example}

\begin{example}
For $p=23 \equiv 11 \pmod{12}$ and $k=1$, the parity check matrix of the code $\mathcal{C}(\F_{23}^2, H_{23})$ is
\begin{align*}
\begin{pmatrix}
1 & 2 & 3 & 4 & 5 & 6 & 7 & 8 & 9 & 10 & 11 \\
1 & 8 & 4 & 18 & 10 & 9 & 21 & 6 & 16 & 11 & 20
\end{pmatrix},
\end{align*}
while the parity check matrix of the Mesnager-Qi-Tang code in Theorem~\ref{thm-MTQ2018} (2) is  
\begin{align*}
\begin{pmatrix}
1 & 2 & 3 & 4 & 5 & 6 & 7 & 8 & 9 & 10 & 11 \\
1 & 12 & 8 & 6 & 14 & 4 & 10 & 3 & 18 & 7 & 21
\end{pmatrix}.
\end{align*}
\end{example}

\begin{example}
For $p=5$ and $k=2$, the parity check matrix of the code $\mathcal{C}(\F_{5}^{4}, H_{5^2})$ is
\begin{align*}
\left(\begin{array}{rrrrrrrrrrrr}
1 & 3 & 1 & 0 & 3 & 3 & 2 & 1 & 2 & 0 & 1 & 1 \\
0 & 4 & 4 & 1 & 3 & 1 & 0 & 3 & 3 & 2 & 1 & 2 \\
1 & 0 & 2 & 0 & 4 & 0 & 3 & 0 & 1 & 0 & 2 & 0 \\
0 & 1 & 0 & 2 & 0 & 4 & 0 & 3 & 0 & 1 & 0 & 2
\end{array}\right).
\end{align*}
Here, we express $\F_{25}$ as $\F_5[\sqrt{2}]$ by taking the irreducible polynomial $x^2-2$, and choose $3+4\sqrt{2}$ as a primitive element of $\F_{25}$.
\end{example}

The proof of Theorem~\ref{thm-main} is based on careful verifications for the equivalent condition in Lemma~\ref{lem-generic} (2), which are completed by the following series of technical lemmas.

\begin{lemma}
\label{lem-count}
For any prime $p \geq 5$ and power $q=p^k$ with $k\geq 1$, 
we have 
\begin{align*}
    \#\bigcup_{i=0}^{2} H_q^{(i)}=2\cdot \Bigl(\frac{q-1}{2} \Bigr)^2+2\cdot \Big(\frac{q-1}{2} \Big)+ 1
\end{align*}
    
\end{lemma}

\begin{proof}
First we show that $\#H_q^{(2)}=\frac{(q-1)^2}{2}+1$.
It suffices to count $(a, b)\in \F_q^2$ such that there are $x, y\in \F_q^*$ with
\begin{align}
\label{eq-count}
  \begin{cases}
    a=x+y  \\
    b=x^3+y^3
  \end{cases}
\end{align}
Note that $a=0$ if and only if $y=-x$, equivalent to $b=0$, implying that there are exactly $\frac{q-1}{2}$ unordered pair $\{x, y\}$ satisfying \eqref{eq-count} with $(a, b)=(0, 0)$. 
If $a\neq 0$, then $(a, b)$ corresponds to an unordered pair $\{x, y\}$ satisfying \eqref{eq-count} and $x+y\neq 0$ by a one-to-one mapping. 
Indeed, by \eqref{eq-count}, we have $xy=\frac{a^3-b}{3a}$, and hence $x, y$ must be roots of the quadratic equation $t^2-at+\frac{a^3-b}{3a}=0$. 
The total number of unordered pairs $\{x,y\}$ from $\mathbb{F}_q^*$ (allowing $x=y$) is equivalent to choosing 2 elements from $q-1$ elements with replacement, which is $\binom{(q-1)+2-1}{2} = \binom{q}{2} = \frac{q(q-1)}{2}$.
Thus, the number of unordered pairs with $x+y \neq 0$ is:
\begin{align*}
    \frac{q(q-1)}{2} - \frac{q-1}{2} = \frac{(q-1)^2}{2}.
\end{align*}
Hence, we conclude that $\#H_q^{(2)} = \frac{(q-1)^2}{2}+1$ by adding $(0,0)$.
The proof completes after showing that $H_q \cap H_q^{(2)}=\emptyset$ as $H_q^{(0)} \subset H_q^{(2)}$ and $\#H_q^{(1)}=\#H_q=q-1$.
Suppose that there exists $(a, a^3) \in H_q$ such that there are $x, y \in \F_q^*$ with 
\begin{align*}
  \begin{cases}
    a=x+y  \\
    a^3=x^3+y^3
  \end{cases}
\end{align*}
This implies that $3axy=0$, which is impossible as $x, y, a \neq 0$ and $p\geq 5$.
\end{proof}

\begin{lemma}
\label{lem-abs}
Let $p \geq 5$ be a prime. Suppose that a power $q=p^k$ with $k\geq 1$. 
Then the following polynomial $F_{a,b}(x,y)$ over $\F_q$ is absolutely irreducible unless $b=a^3$.
    \begin{align*}
   F_{a,b}(x,y) \coloneq xy(x+y) &-a(x+y)^2+a^2(x+y)+\frac{b-a^3}{3}
\end{align*}
If $b=a^3$, then for any $a\in \F_q$, we have
\begin{align*}
    F_{a, a^3}(x,y)=(x+y)(x-a)(y-a).
\end{align*}
\end{lemma}

\begin{proof}
As the statement for the case $b=a^3$ is trivial, 
suppose $b \neq a^3$ and $F_{a,b}$ is not absolutely irreducible.
Then, $F_{a,b}$ must have a linear factor which is either $x+y-c$, $x-c$ or $y-c$ for some constant $c$. 
If $x+y-c$ is a factor of $F_{a,b}$, then substituting $x+y=c$ into $F_{a,b}(x,y)=0$ yields that $cxy-ac^2+a^2c+\frac{b-a^3}{3}=0$ holds for any $x$ and $y$. Hence we have $b=a^3$, a contradiction.
If $x-c$ is a factor of $F_{a,b}$, then it holds for any $y$ that
$cy(c+y)-a(c+y)^2+a^2(c+y)+\frac{b-a^3}{3}=0$.
Hence we have $c=a$ and $b=a^3$, which cannot occur by the assumption.
By the same discussion, $y-c$ cannot be a factor of $F_{a,b}$, completing the proof.
\end{proof}

\begin{lemma}
\label{lem-covering}
Let $p \geq 5$ be a prime. Suppose that a power $q=p^k$ with $k\geq 1$ satisfies $q\geq 14$.
Then we have 
\begin{align*}
    \bigcup_{i=0}^{3} H_q^{(i)}=\Gamma
\end{align*}
\end{lemma}

\begin{proof}
As $\{(0,0)\}=H_q^{(0)}$,
it suffices to prove that for any $(a, b) \in \F_q^2 \setminus \{(0,0)\}$, there are $x, y, z \in \F_q^*$ such that 

\begin{align*}
  \begin{cases}
    a=x+y+z  \\
    b=x^3+y^3+z^3
  \end{cases}
\end{align*}

Then, by substituting $z=a-x-y$ into the second equation,
we have the cubic curve $C_{a,b}:F_{a,b}(x,y)=0$ of genus one.
%
Hence, desired elements $x, y, z \in \F_q^*$ exist if there exists a point $(x,y)$ on the curve $C_{a,b}$ with $x, y \neq 0$ and $x+y\neq a$.

If $b \neq a^3$, then $F_{a,b}$ is absolutely irreducible by Lemma~\ref{lem-abs}. 
In addition, on the projective space $\mathbb{P}^2_{\F_q}$ of dimension $2$, the projective points $(X:Y:Z)$ with $Z=0$ are $(1:0:0), (0:1:0)$ and $(1:-1:0)$ which must be excluded.
Thus, by the Hasse-Weil bound (e.g. \cite{S1986}),
the number of points of $C_{a,b}$ is at least $q-2\sqrt{q}-2$.
Note that $C_{a,b}$ cannot intersect with the line $x+y=a$ if $b\neq a^3$ by the proof of Lemma~\ref{lem-abs}, and each of the lines $x=0$ and $y=0$ intersects with $C_{a,b}$ at no more than two points.
Hence, if $b \neq a^3$, there exists a point $(x,y)$ with $x, y \neq 0$ and $x+y\neq a$ if $(q-2\sqrt{q}-2)-2\cdot2>0$, which holds as $q\geq 14$.

If $b=a^3$, as shown in Lemma~\ref{lem-abs}, the curve $C_{a,b}$ contains the line $x+y=0$. 
Since $a\neq 0$ holds by the assumption, all points $(t,-t)$, $t \in \F_q^*$, on $C_{a,b}$ are desired ones.

These complete the proof.
\end{proof}

\begin{remark}
For characteristic $2$, a result analogous to Lemma~\ref{lem-covering} has been established in \cite[Theorem~10.2.3]{CHLL1997}. While their proof is purely combinatorial, our approach for characteristic $p \geq 5$ crucially relies on techniques from algebraic geometry.
\end{remark}

Now we are ready to prove Theorem~\ref{thm-main}.

\begin{proof}[Proof of Theorem~\ref{thm-main}]
Notice that $\F_q^2$ is a vector space over $\F_p$ with $\dim_{\F_p}(\F_q^2)=2k$. 
Also, $H_q$ clearly spans $\F_q^2$ by Lemma~\ref{lem-covering}.
By taking $n=\frac{q-1}{2}$,  
all the conditions of Lemma~\ref{def-generic} are satisfied by Lemmas~\ref{lem-count} and \ref{lem-covering}, 
completing the proof.   
\end{proof}

\section{The Mesnager-Tang-Qi codes and abelian Ramanujan graphs}
\label{sect-MTQ}
Now recall the explicit $2$-quasi-perfect Lee codes by Mesnager, Tang and Qi in \cite{MTQ2018}. 
For an odd prime $p$, let $q=p^k$ where $k$ is an integer.
Then, the field extension $\F_{q^2}$ of $\F_q$ is expressed as  $\F_{q^2}=\F_q[\sqrt{\delta}]=\{x+y\sqrt{\delta} \mid x, y \in \F_q\}$ where $\delta \in \F_q^*$ is a non-square.
For $z=x+y\sqrt{\delta} \in \F_{q^2}$, let $\overline{z}\coloneq x-y\sqrt{\delta}$ and
$\mathcal{N}(z)\coloneq z\cdot \overline{z}=x^2-\delta y^2$.
Then define 
\begin{align*}
    &H_{+}\coloneq \{z \in \F_{q^2} \mid \mathcal{N}(x+y\sqrt\delta)=1 \},\\
    &H_{-}\coloneq \{(x,y) \in \F_{q}^2 \mid xy=1\}.
\end{align*}
Clearly, both $H_{+}$ and $H_{-}$ are subsets of $\F_q^2$.

\begin{theorem}[Theorem~8 in \cite{MTQ2018}]
\label{thm-MTQ2018}
For an odd prime $p$, let $q=p^k$ where $k$ is an integer.
Then we have the following.
\begin{enumerate}
    \item[$(1)$] If either $p\equiv 1, 7 \pmod{12}$, or $p\geq 5$ and $k$ is even, then the code $\mathcal{C}(\F_q^2, H_{+})$ is a $2$-quasi-perfect $p$-ary code of length $\frac{q+1}{2}$ and dimension $\frac{q+1}{2}-2k$.

     \item[$(2)$] If $p\equiv 5, 11 \pmod{12}$ and $q>12$, then the code $\mathcal{C}(\F_q^2, H_{-})$ is a $2$-quasi-perfect $p$-ary code of length $\frac{q-1}{2}$ and dimension $\frac{q-1}{2}-2k$.
\end{enumerate}

\end{theorem}



The first observation here is that the code in Theorem~\ref{thm-MTQ2018} (1) is in fact obtained by the generating set of the following abelian Ramanujan graph, called {\it Li's graph}.

\begin{definition}[\cite{LK1992}]
For any odd prime power $q$,
let $N_2:=\{z \in \F_{q^2} \mid \Nsq(z)=1\}$.
Then, {\it Li's graph} 
over $\F_{q^2}$ is the Cayley graph $Cay(\F_{q^2}^+, N_2)$.  
\end{definition}

It is known in graph theory and number theory that Li's graphs are Ramanujan graphs by Deligne's bound for Kloosterman sums over $\F_{q^2}$ (\cite{LK1992, D2023}). 

As observed in \cite{BKS2016} for $q=p \equiv 3 \pmod{4}$, one can easily verify:

\begin{proposition}
\label{prop-Li}
For any odd prime power $q$, it holds that $N_2=H_{+}$. 
Hence, the code $\mathcal{C}(\F_q^2, H_{+})$ coincides with $\mathcal{C}(\F_q^2, N_2)$.
\end{proposition}

\begin{proof}
It suffices to prove that $\mathcal{N}$ coincides with $\Nsq$.
For any $z=x+y\sqrt{\delta} \in \F_{q^2}$, we have
\begin{align*}
    \Nsq(z)
    &=(x+y\sqrt{\delta})(x+y\sqrt{\delta})^q\\
    &=(x+y\sqrt{\delta})(x^q+y^q (\sqrt{\delta})^q)\\
    &=(x+y\sqrt{\delta})(x+y  (\sqrt{\delta})^q)\\
    &=(x+y\sqrt{\delta})(x+y (-\sqrt{\delta}))\\
    &=z\cdot \overline{z}=\mathcal{N}(z),
\end{align*}
where $ (\sqrt{\delta})^q=-\sqrt{\delta}$ as $\delta^{q-1}=1$ and $\delta^{\frac{q-1}{2}}\neq 1$ as $\delta$ is non-square.
\end{proof}

The second observation is that all codes in Theorem~\ref{thm-MTQ2018} are unified by the concept of {\it finite Euclidean graphs}.

\begin{definition}[\cite{BST2009, MMST1996}]
For any odd prime power $q$, let $Q(x, y)$ be a non-degenerate quadratic form over $\mathbb{F}_q$, and define $E_Q \coloneq \{(x,y) \in \mathbb{F}_q^2 \mid Q(x,y)=1\}$.
The {\it unit finite Euclidean graph} associated with $Q$ is the Cayley graph $\mathrm{Cay}(\mathbb{F}_{q^2}^+, E_Q)$.
\end{definition}

Let $Q_+(x,y) \coloneq x^2-\delta y^2$ and $Q_-(x,y) \coloneq xy$, where $\delta \in \mathbb{F}_q^*$ is a non-square element.
Note that $Q_+$ and $Q_-$ are inequivalent, and any non-degenerate quadratic form $Q$ over $\mathbb{F}_q$ is equivalent to either $Q_+$ or $Q_-$.
Recall that two quadratic forms $Q_1(\mathbf{x})$ and $Q_2(\mathbf{x})$ are equivalent if there exists $A \in \mathrm{GL}(2, \mathbb{F}_q)$ such that $Q_2(\mathbf{x})=Q_1(A\mathbf{x})$ for all $\mathbf{x} \in \mathbb{F}_q^2$.
Consequently, the finite Euclidean graph $\mathrm{Cay}(\mathbb{F}_{q^2}^+, E_Q)$ is isomorphic to either $G_+ \coloneq \mathrm{Cay}(\mathbb{F}_{q^2}^+, E_{Q_+})$ or $G_- \coloneq \mathrm{Cay}(\mathbb{F}_{q^2}^+, E_{Q_-})$. The isomorphism is induced by the linear transformation providing the equivalence between $Q$ and $Q_{\pm}$.
It is well-known that both $G_+$ and $G_-$ are (almost) Ramanujan graphs. This follows from the character tables of affine association schemes and Deligne's bound for Kloosterman sums (e.g., \cite{BST2009, MMST1996, D2023}).

We summarize the connection to the Mesnager-Qi-Tang codes as follows:

\begin{proposition}
\label{prop-euclid}
For any odd prime power $q$, we have $E_{Q_{\pm}}=H_{\pm}$, and hence, the codes $\mathcal{C}(\mathbb{F}_q^2, H_{\pm})$ in Theorem~\ref{thm-MTQ2018} coincide with $\mathcal{C}(\mathbb{F}_q^2, E_{Q_{\pm}})$.
Furthermore, for any non-degenerate quadratic form $Q$, the code $\mathcal{C}(\mathbb{F}_q^2, E_Q)$ coincides with either $\mathcal{C}(\mathbb{F}_q^2, E_{Q_{+}})$ or $\mathcal{C}(\mathbb{F}_q^2, E_{Q_{-}})$ up to coordinate permutations and sign changes (of coordinates), depending on whether $Q$ is equivalent to $Q_+$ or $Q_-$ in $\mathbb{F}_q$.
\end{proposition}

This proposition, combined with Theorem~\ref{thm-MTQ2018}, implies that for any non-degenerate quadratic form $Q$, the set $E_Q$ essentially yields only the Mesnager-Tang-Qi codes.
The second claim of Proposition~\ref{prop-euclid} follows immediately from the Lemma below.

\begin{lemma}
Let $Q$ and $Q'$ be two non-degenerate quadratic forms over $\mathbb{F}_q$. If $Q$ and $Q'$ are equivalent, then $\mathcal{C}(\mathbb{F}_q^2, E_Q)$ coincides with $\mathcal{C}(\mathbb{F}_q^2, E_{Q'})$ up to coordinate permutations and sign changes.
\end{lemma}

\begin{proof}
By assumption, there exists $A \in \mathrm{GL}(2, \mathbb{F}_q)$ such that $Q'(\mathbf{x}) = Q(A\mathbf{x})$ for all $\mathbf{x} \in \mathbb{F}_q^2$. This implies
\begin{align*}
    E_{Q'} &= \{ \mathbf{x} \in \mathbb{F}_q^2 \mid Q'(\mathbf{x}) = 1 \} \\
    &= \{ \mathbf{x} \in \mathbb{F}_q^2 \mid Q(A\mathbf{x}) = 1 \} \\
    & = \{ A^{-1}\mathbf{y} \mid \mathbf{y} \in E_Q \}.
\end{align*}

Let $E_Q = S \sqcup (-S)$ where $S = \{ \beta_1, \dots, \beta_n \}$ is a fixed set of representatives.
Then, the set $S' = \{ \beta'_1, \dots, \beta'_n \}$ with $\beta'_i = A^{-1}\beta_i$ forms a valid set of representatives for $E_{Q'}$.
Note that if $\mathcal{C}(\mathbb{F}_q^2, E_{Q'})$ is defined by another set of representatives $T'$, then $S'$ and $T'$ differ only by the ordering of elements and the signs of coordinates (since $E_{Q'}$ is inverse-closed).

Let $q=p^k$, where $p$ is an odd prime and $k\geq 1$.
Viewing the ambient space $\mathbb{F}_q^2$ as a $2k$-dimensional vector space over $\mathbb{F}_p$, the map $A^{-1}$ acts as an $\mathbb{F}_p$-linear transformation, denoted by $f_A^{-1}$.
A vector $\mathbf{c} = (c_1, \dots, c_n) \in \mathbb{F}_p^n$ is a codeword of $\mathcal{C}(\mathbb{F}_q^2, E_Q)$ if and only if
\[
    \sum_{i=1}^n c_i \beta_i = \mathbf{0} \quad (\text{in } \mathbb{F}_p^{2k}).
\]
Applying $f_A^{-1}$ to both sides yields
\[
    f_A^{-1} \left( \sum_{i=1}^n c_i \beta_i \right) = \sum_{i=1}^n c_i (f_A^{-1}(\beta_i)) = \sum_{i=1}^n c_i \beta'_i = \mathbf{0}.
\]
This implies that $\mathbf{c}$ is a codeword for the code defined by representatives $\{\beta'_i\}$. Since the representatives of $E_{Q'}$ are unique only up to signs and ordering, the result follows.
\end{proof}

\begin{example}
Let $Q(x, y)=x^2+y^2$.
\begin{itemize}
    \item If $q \equiv 1 \pmod{4}$, then $-1$ is a square in $\mathbb{F}_q$, so let $i^2=-1$. $Q$ is equivalent to $Q_{-}$ via the transformation:
$$
Q(x, y)=(x+iy)(x-iy)=Q_-(u,v),
$$
where $u=x+iy$ and $v=x-iy$.
    Thus, $\mathcal{C}(\mathbb{F}_q^2, E_Q)$ coincides with $\mathcal{C}(\mathbb{F}_q^2, E_{Q_-})$ up to coordinate permutations and sign changes.
    \item If $q \equiv 3 \pmod{4}$, then $-1$ is non-square. $Q$ is equivalent to $Q_{+}$ by setting $\delta = -1$:
$$
Q(x, y)=x^2-(-1)y^2=Q_+(x,y).
$$
Thus, $\mathcal{C}(\mathbb{F}_q^2, E_Q)$ coincides with $\mathcal{C}(\mathbb{F}_q^2, E_{Q_+})$.
\end{itemize}
\end{example}

\begin{example}
Let $Q(x, y) = x^2+xy+y^2$. The discriminant of $Q$ is $-3$.
Completing the square, we have
\[
Q(x, y) = \left(x + \frac{y}{2}\right)^2 + \frac{3}{4}y^2.
\]
Let $u = x + \frac{y}{2}$ and $v = y$. Then $Q(x,y)$ is equivalent to the diagonal form $u^2 + \frac{3}{4}v^2$.

\begin{itemize}
    \item If $-3$ is a square in $\mathbb{F}_q$, then 
    $\frac{3}{4}$ can be written as $-s^2$ for some $s \in \mathbb{F}_q^*$. Thus, $Q$ is equivalent to $u^2 - s^2 v^2 = (u-sv)(u+sv)$, which is equivalent to $Q_-(x, y) = xy$.
    In this case, $\mathcal{C}(\mathbb{F}_q^2, E_Q)$ coincides with $\mathcal{C}(\mathbb{F}_q^2, E_{Q_-})$ up to coordinate permutations and sign changes.

    \item If $-3$ is a non-square in $\mathbb{F}_q$, then 
    $-\frac{3}{4}$ is a non-square (since $\frac{1}{4}$ is a square). Let $\delta = -\frac{3}{4}$. 
    Then $Q$ is equivalent to $u^2 - \delta v^2$, which is exactly $Q_+(u, v)$.
    In this case, $\mathcal{C}(\mathbb{F}_q^2, E_Q)$ coincides with $\mathcal{C}(\mathbb{F}_q^2, E_{Q_+})$ up to coordinate permutations and sign changes.
\end{itemize}
\end{example}

\section{Concluding remarks}
\label{sect-conc}
\begin{itemize}
\item It would be of significant interest to investigate whether other abelian (almost) Ramanujan graphs harbor generating sets that yield new families of quasi-perfect Lee codes. 
Such discoveries would further solidify the profound connection between the theory of Lee codes and spectral graph theory.

\item The 2-quasi-perfect Lee codes presented in this paper, as well as those in \cite{MTQ2018}, are derived from abelian (almost) Ramanujan graphs that are {\it spectrally indistinguishable} from random graphs of the same edge density. 
That is, their normalized eigenvalue distributions converge to the semicircle distribution, a hallmark of typical random graphs. 
In this sense, these graphs exhibit strong pseudo-randomness. 
Conversely, the property of 2-quasi-perfectness imposes rigid structural regularities on the underlying Cayley graph, as dictated by Lemma~\ref{lem-generic}. 
Notably, the strictly deterministic condition in Lemma~\ref{lem-generic} (2) is not a property that random graphs possess with high probability. 
Therefore, we observe that spectral indistinguishability does not preclude the existence of the highly ordered local structures required for $2$-quasi-perfectness. 
This underscores a fascinating dichotomy: spectral indistinguishability alone is insufficient to guarantee ``local" pseudo-randomness, revealing a structural gap akin to the phenomena discussed in \cite{FFKW2025}.


\end{itemize}


\section*{Acknowledgment}
The author deeply appreciates the anonymous reviewers for their careful reading and constructive comments, which significantly improved the presentation and clarity of this paper.
This work was supported by JSPS KAKENHI Grant Number JP	23K13007.
This work was also supported by Institute of Mathematics for Industry, Joint Usage/Research Center in Kyushu University (FY2025 Short-term Joint Research, Reference No. 2025a037).

\end{document}